\documentclass{article}%
\usepackage{amsfonts}
\usepackage{amsmath}
\usepackage{amssymb}
\usepackage{graphicx}%
\providecommand{\U}[1]{\protect\rule{.1in}{.1in}}
\newtheorem{theorem}{Theorem}

\newenvironment{proof}[1][Proof]{\noindent\textbf{#1.} }{\ \rule{0.5em}{0.5em}}
\begin{document}

\title{Revisiting Simpson's Paradox: \\a statistical misspecification perspective}
\author{Aris Spanos\\Department of Economics, \\Virginia Tech, Blacksburg,\\VA 24061, USA}
\date{May 2016}
\maketitle

\begin{abstract}
The primary objective of this paper is to revisit Simpson's paradox using a
statistical misspecification perspective. It is argued that the reversal of
statistical associations is sometimes spurious, stemming from invalid
probabilistic assumptions imposed on the data. The concept of statistical
misspecification is used to formalize the vague term `spurious results' as
`statistically untrustworthy' inference results. This perspective sheds new
light on the paradox by distingusing between statistically trustworthy vs.
untrustworthy association reversals. It turns out that in both cases there is
nothing counterintuitive to explain or account for. This perspective is also
used to revisit the causal `resolution' of the paradox in an attempt to
delineate the modeling and inference issues raised by the statistical
misspecification perspective. The main arguments are illustrated using both
actual and hypothetical data from the literature, including Yule's
"nonsense-correlations" and the Berkeley admissions study.


\end{abstract}

\newpage

\section{Introduction\vspace*{-0.1in}}

True to "Stigler's Law of Eponymy" (Stigler, 1980), Simpson's paradox has a
long history in statistics going back to Yule's (1903) `spurious' association,
but it is currently credited to Simpson (1951) for reframing it as a
`paradox'; see Blyth (1972).

The paradox seems to have a number of alternative conceptions, and thus, it is
often described interchangeably as a counter-intuitive statistical result
pertaining to:

(a) Statistical associations that reverse themselves, such as "a marginal
association can have a different direction from each conditional association"
(Agresti, 2013).

(b) Either the magnitude or the direction of an association between two
variables is influenced by a third variable, such as "the association between
a pair of variables $\left(  X,Y\right)  $ reverses sign upon conditioning on
a third variable, $Z$.\textquotedblright\ (Pearl, 2014).

(c) Apparent statistical associations that after closer scrutiny of the data
are rendered `spurious' (Yule, 1903).

The recent discussions in statistics have focused on adopting one of the
perspectives (a)-(c), and using actual or hypothetical data to either explain
away the paradox or criticize other proposed `solutions'. The current
dominating view revolves around perspective (b) that differs from (a) in so
far as it emphasizes the causal dimension of conditioning on a confounder; see
Pearl (2009), Spirtes et al. (2000). Armistead (2014) put forward a dissenting
view by arguing that perspective (b) is rather narrow to explain the different
facets of the paradox:

\textsf{\textquotedblleft Simpson's Paradox, like all paradoxes, can be
defined as an apparent contradiction that may contain more than one
truth.\textquotedblright\ (p. 6)}

A strong case can be made that Simpson's paradox has different dimensions that
are often conflated or ignored in the literature. As argued by Wasserman (2004):

\textsf{\textquotedblleft Simpson's paradox is a puzzling phenomenon that is
discussed in most statistics texts. Unfortunately, most explanations are
confusing (and in some cases incorrect).\textquotedblright\ (p. 259)}

The primary aim of this paper is to shed light on the different conceptions of
the paradox by bringing out the similarities and differences between
perspectives (a)-(c). The key is provided by Yule's idea of `spuriousness' in
(c). Beginning with Yule (1903), the problem of `fictitious' associations and
`spurious' correlations was a recurring theme in Yules papers that culminated
in Yule (1926) on "nonsense-correlations". Although he shed some light on the
issues involved, he did not succeed in establishing a direct link between
spurious associations and invalid probabilistic assumptions for reasons to be
discussed in the sequel. The notion of statistical misspecification can be
used to formalize the term `spurious' as `statistically untrustworthy'
results, stemming from unreliable inference procedures. This enables one to
delineate between two distinct cases of association reversal:

\hspace*{-0.25in}Case 1. The reversal is statistically trustworthy due to
statistical adequacy.

\hspace*{-0.25in}Case 2. The reversal is statistically untrustworthy due to
statistical misspecification.

It turns out that the statistical misspecification perspective suggests that
in both cases there is nothing counterintuitive to explain.

In section 2, we discuss the case where the reversal is statistically
trustworthy due to the fact that the statistical models involved are
\textit{statistically adequate}: the invoked probabilistic assumptions are
valid for the particular data. When this is not the case, the inference
results are likely to be statistically untrustworthy (spurious). This is
discussed in section 3 using two empirical examples that bring out the
distinction between statistical and substantive misspecification. The
statistical misspecification argument is illustrated further in section 4
using several widely discussed examples of the paradox. In section 5, we
revisit the causal `resolution' of the paradox in an attempt to delineate the
modeling and inference issues raised by the statistical misspecification
perspective. \vspace*{-0.15in}

\section{Marginal vs. conditional associations\vspace*{-0.1in}}

Consider the case of a Linear Regression (LR) model:\vspace*{-0.1in}%
\begin{equation}%
\begin{array}
[c]{c}%
y_{t}\mathit{=}\beta_{0}+\beta_{1}x_{1t}+\beta_{2}x_{2t}\mathit{+}%
u_{t},\medskip\\
(u_{t}\mathit{\mid}X_{1t}\mathit{=}x_{1t},X_{2t}\mathit{=}x_{2t}%
)\mathit{\backsim}\text{\textsf{NIID}}(0,\sigma_{u}^{2}),\ t\mathit{\in
}\mathbb{N},
\end{array}
\vspace*{-0.07in}\label{rm1}%
\end{equation}
where `NIID' stands for `Normal, Independent and Identically Distributed'. It
is often insufficiently appreciated that the error assumptions imply a
particular statistical parameterization for the unknown parameters
$\mathbf{\theta}\mathit{:=}\left(  \beta_{0},\beta_{1},\beta_{2},\sigma
^{2}\right)  $ in terms of the moments of the observable process
$\{\mathbf{Z}_{t}\mathit{:=}(y_{t}\mathit{,}X_{1t},X_{2t}),\ t\mathit{\in
}\mathbb{N}\}$ underlying data $\mathbf{Z}_{0}$ (see Appendix). Alternatively,
one can derive the parameterization directly using the joint distribution of
the observable random variables involved:\vspace*{-0.1in}%
\begin{equation}%
\begin{array}
[c]{c}%
\left(
\begin{array}
[c]{c}%
y_{t}\\
X_{1t}\\
X_{2t}%
\end{array}
\right)  \backsim\text{ \textsf{NIID}}\left(  \left(
\begin{array}
[c]{c}%
\mu_{1}\\
\mu_{2}\\
\mu_{3}%
\end{array}
\right)  ,\left(
\begin{array}
[c]{ccc}%
\sigma_{11} & \sigma_{12} & \sigma_{13}\\
\sigma_{12} & \sigma_{22} & \sigma_{23}\\
\sigma_{13} & \sigma_{23} & \sigma_{33}%
\end{array}
\right)  \right)
\end{array}
\vspace*{-0.1in}\label{j}%
\end{equation}
In this case, the regression and skedastic functions take the form:\vspace
*{-0.1in}%
\[%
\begin{array}
[c]{c}%
E(y_{t}\mathit{\mid}X_{1t}\mathit{=}x_{1t},X_{2t}\mathit{=}x_{2t}%
)\mathit{=}\beta_{0}\mathit{+}\beta_{1}x_{1t}\mathit{+}\beta_{2}%
x_{2t},\ Var(y_{t}\mathit{\mid}X_{1t}\mathit{=}x_{1t},X_{2t}\mathit{=}%
x_{2t})\mathit{=}\sigma^{2},
\end{array}
\vspace*{-0.1in}%
\]
where the parameterizations of $\mathbf{\theta}\mathit{:=}\left(  \beta
_{0},\beta_{1},\beta_{2},\sigma^{2}\right)  $ are (table 1):%
\begin{equation}%
\begin{array}
[c]{c}%
\beta_{0}\mathit{=}\mu_{1}\mathit{-}\beta_{1}\mu_{2}\mathit{-}\beta_{2}\mu
_{3},\ \beta_{1}\mathit{=}\frac{(\sigma_{12}\sigma_{33}-\sigma_{13}\sigma
_{23})}{(\sigma_{22}\sigma_{33}-\sigma_{23}^{2})},\ \beta_{2}\mathit{=}%
\frac{(\sigma_{13}\sigma_{22}-\sigma_{12}\sigma_{23})}{(\sigma_{22}\sigma
_{33}-\sigma_{23}^{2})}%
\end{array}
\vspace*{-0.1in}\label{b}%
\end{equation}%
\begin{equation}%
\begin{array}
[c]{cl}%
\sigma_{u}^{2} & \mathit{=}\sigma_{11}\mathit{-}\sigma_{12}\left(
\frac{\sigma_{12}\sigma_{33}-\sigma_{13}\sigma_{23}}{\sigma_{22}\sigma
_{33}-\sigma_{23}^{2}}\right)  \mathit{-}\sigma_{13}\left(  \frac{\sigma
_{13}\sigma_{22}-\sigma_{12}\sigma_{23}}{\sigma_{22}\sigma_{33}-\sigma
_{23}^{2}}\right)  \mathit{=}\sigma_{11}\mathit{-}\sigma_{12}\beta
_{1}\mathit{-}\sigma_{13}\beta_{2}%
\end{array}
\label{s}%
\end{equation}
These results offer the key to elucidating perspectives (a)-(b) on Simpson's paradox.

\textbf{Perspective (a) on Simpson's paradox}. The correlation between $y_{t}
$ and $X_{1t}\ $\newline$(\rho_{12}$\textit{$=$}$\frac{\sigma_{12}}%
{\sqrt{\sigma_{11}\sigma_{22}}}),$ is positive ($\rho_{12}\mathit{>}0$), but
the coefficient $\beta_{1}$ in (\ref{rm1}) is negative ($\beta_{1}\mathit{<}%
0$).\medskip

\textbf{Is this reversal of association possible, and under what
circumstances?}

In light of the parameterization of $\beta_{1}$ in (\ref{b}), its numerator is
negative when:$\vspace*{-0.07in}$%
\[%
\begin{array}
[c]{c}%
\left[  (\sigma_{12}\sigma_{33}\mathit{-}\sigma_{13}\sigma_{23})<0\right]
\rightarrow\left[  \frac{\sigma_{13}\sigma_{23}}{\sigma_{33}}>\sigma
_{12}\right]
\end{array}
\vspace*{-0.07in}%
\]
Multiplying both terms in the last expression by $1/\sqrt{\sigma_{11}%
\sigma_{22}},$ yields:$\vspace*{-0.07in}$
\[%
\begin{array}
[c]{c}%
\frac{\sigma_{13}\sigma_{23}}{\sigma_{33}\sqrt{\sigma_{11}\sigma_{22}}%
}\mathit{=}\rho_{13}\rho_{23}>\rho_{12}\mathit{=}\frac{\sigma_{12}}%
{\sqrt{\sigma_{11}\sigma_{22}}},
\end{array}
\vspace*{-0.05in}%
\]
where $Corr(Y_{t},X_{2t})\mathit{:=}\rho_{13}$\textit{$=$}$\frac{\sigma_{13}%
}{\sqrt{\sigma_{33}\sigma_{11}}}$ and $Corr(X_{1t},X_{2t})\mathit{:=}\rho
_{23}$\textit{$=$}$\frac{\sigma_{23}}{\sqrt{\sigma_{22}\sigma_{33}}}$.\medskip

Hence, $\rho_{12}\mathit{>}0$ and $\beta_{1}\mathit{<}0$ occur when the
following conditions hold:\medskip

(i) the correlation coefficients $\rho_{13}$ and $\rho_{23}$ have the
\textit{same sign},

(ii) the product of $\rho_{13}$ and $\rho_{23}$ is greater than $\rho_{12} $,
i.e.\ $\rho_{13}\mathit{\cdot}\rho_{23}\mathit{>}\rho_{12},$ and

(iii) the determinant of the correlation matrix of $\mathbf{Z}_{t}$ is
positive:$\vspace*{-0.07in}$%
\[%
\begin{array}
[c]{c}%
Corr(\mathbf{Z}_{t})\mathit{=}1\mathit{-}\rho_{12}^{2}\mathit{-}\rho_{13}%
^{2}\mathit{-}\rho_{23}^{2}\mathit{+}2\rho_{12}\rho_{13}\rho_{23}\mathit{>}0.
\end{array}
\vspace*{-0.07in}%
\]
Condition (iii) ensures that $f(y_{t},x_{1t},x_{2t};\mathbf{\phi})$ in
(\ref{j}) is proper, giving rise to a well-defined conditional distribution
$f(y_{t}\mathit{\mid}x_{1t},x_{2t};\mathbf{\varphi})$; see Spanos and McGuirk
(2002).\vspace*{-0.15in}

\subsection{\textbf{Example 1}. \textbf{Correlations vs. partial
correlations}\vspace*{-0.1in}}

Assuming $\sigma_{11}\mathit{=}\sigma_{22}\mathit{=}\sigma_{33}\mathit{=}1,$
without any loss of generality, let the relevant correlations be: $\left(
\rho_{12},\rho_{13},\rho_{23}\right)  \mathit{=}\left(  .5,\pm.7,\pm.8\right)
,$ which satisfy (i)-(iii) above.\medskip

(a) For values $\left(  \rho_{13},\rho_{23}\right)  \mathit{=}\left(
.7,.8\right)  $: $\beta_{1}\mathit{=}-.167,\ \beta_{2}\mathit{=}%
.833,\ \sigma^{2}\mathit{=}.5$\medskip

(b) For values $\left(  \rho_{13},\rho_{23}\right)  \mathit{=}\left(
-.7,-.8\right)  $: $\beta_{1}\mathit{=}-.167,\ \beta_{2}\mathit{=}%
-.833,\ \sigma^{2}\mathit{=}.5$\medskip

Note that the sign of $\beta_{2}$ reflects the common sign of $\left(
\rho_{13},\rho_{23}\right)  .$ In light of these results, it is clear that
there is nothing paradoxical, or surprising, about the reversal of sign
between the simple correlation $\rho_{12}\mathit{>}0$ [stemming from the joint
distribution $f(y_{t}\mathit{,}x_{1t};\mathbf{\varphi}_{1})$], and the
regression coefficient $\beta_{1}\mathit{<}0$ [stemming from the conditional
distribution $f(y_{t}\mathit{\mid}x_{1t},x_{2t};\mathbf{\varphi}_{2})$]. This
reversal is due to the conditions (i)-(iii) above, which are easily testable
in practice; see Spanos (2006b).

It is well-known that there is a direct connection between $\rho_{12}$ and the
regression coefficient of $x_{1t}$ in the context of the simple linear
regression:\vspace*{-0.08in}%
\begin{equation}%
\begin{array}
[c]{c}%
y_{t}\mathit{=}\alpha_{0}\mathit{+}\alpha_{1}x_{1t}+\varepsilon_{t},\medskip\\
(\varepsilon_{t}\mathit{\mid}X_{1t}\mathit{=}x_{1t})\mathit{\backsim
}\text{\textsf{NIID}}(0,\sigma_{\varepsilon}^{2}),\ t\mathit{\in}\mathbb{N},
\end{array}
\vspace*{-0.08in}\label{rm2}%
\end{equation}
whose implicit statistical parameterization of $\mathbf{\phi}\mathit{:=}%
\left(  \alpha_{0},\alpha_{1},\sigma_{1}^{2}\right)  $ is:\vspace*{-0.08in}%
\[%
\begin{array}
[c]{c}%
\alpha_{0}\mathit{=}\mu_{1}\mathit{-}\alpha_{1}\mu_{2},\ \alpha_{1}%
\mathit{=}\frac{\sigma_{12}}{\sigma_{22}},\ \sigma_{\varepsilon}^{2}%
\mathit{=}\sigma_{11}\mathit{-}\frac{\sigma_{12}^{2}}{\sigma_{22}}.
\end{array}
\vspace*{-0.08in}%
\]
This is because $\rho_{12}$ is a scaled $(\sqrt{\sigma_{22}}/\sqrt{\sigma
_{11}})\mathit{>}0)$ reparameterization of $\alpha_{1}$:\vspace*{-0.08in}%
\begin{equation}%
\begin{array}
[c]{c}%
\rho_{12}\mathit{=}\frac{\sqrt{\sigma_{22}}}{\sqrt{\sigma_{12}}}\alpha_{1}.
\end{array}
\vspace*{-0.12in}\label{cr}%
\end{equation}
In the above numerical example, $\alpha_{1}\mathit{=}.5$ and\ $\beta
_{1}\mathit{=}-.167,$ confirming the sign reversal. This implies that one can
consider the question of association reversal by comparing the inference
results in (\ref{rm1}) and (\ref{rm2}).

In conclusion, it is very important to emphasize that in the above example,
both LR models, (\ref{rm1}) and (\ref{rm2}), are assumed to be statistically
adequate: their probabilistic assumptions are valid. In the case of real data
on $\left(  y_{t}\mathit{,}x_{1t},x_{2t}\right)  ,\ t\mathit{=}1,2,...,n$, one
needs to establish the statistical adequacy of both models using comprehensive
misspecification testing. What are the probabilistic assumptions that need to
hold for data $\mathbf{Z}_{0}$?\vspace*{-0.15in}

\section{Spurious (statistically untrustworthy) results\vspace*{-0.1in}}

In this section we bring out more explicitly the probabilistic assumptions
comprising the Linear Regression (LR) model with a view to illustrate the role
of statistical misspecification in shedding light on the various aspects of
Simpson's paradox.\vspace*{-0.15in}

\subsection{The statistical misspecification perspective\vspace*{-0.08in}}

Traditionally, the probabilistic assumptions underlying the Linear Regression
(LR) model are specified in terms of the error term; see Appendix. It turns
out, however, that such specifications are often incomplete and sometimes
include non-testable assumptions. Table 1 specifies the LR, generically
defined by: in terms of the Statistical Generating Mechanism (GM) and
assumptions [1]-[5] that constitute a complete, internally consistent and
testable set of assumptions in terms of the observable process $\{\left(
y_{t}\mathit{\mid}\mathbf{X}_{t}\mathit{=}\mathbf{x}_{t}\right)
,\ t\mathbf{\mathit{\in}}\mathbb{N}\}$ underlying the data\ $\mathbf{Z}%
_{0}\mathit{:=}\{(y_{t},\mathbf{x}_{t}),\ t\mathit{=}1,2,...,n\}$. This
provides a purely probabilistic construal for the notion of a statistical
model, viewed as a particular parameterization of the process $\{\left(
y_{t}\mathit{\mid}\mathbf{X}_{t}\mathit{=}\mathbf{x}_{t}\right)
,\ t\mathbf{\mathit{\in}}\mathbb{N}\}$. Intuitively, the statistical model
comprises the totality of probabilistic assumptions one imposes on the process
$\{\left(  y_{t}\mathit{\mid}\mathbf{X}_{t}\mathit{=}\mathbf{x}_{t}\right)
,\ t\mathbf{\mathit{\in}}\mathbb{N}\}$ with a view to render data
$\mathbf{Z}_{0}$ a `typical' realization thereof. The `typicality' is testable
using thorough misspecification testing; see Spanos (2006a).$\ $%
\[
\hspace*{-0.25in}%
\begin{tabular}
[c]{l}\hline%
\begin{tabular}
[c]{l}%
\textbf{Table 1: Linear Regression Model}%
\end{tabular}
\\\hline\hline
$\overset{\qquad}{%
\begin{tabular}
[c]{ll}%
Statistical GM: & $y_{t}\mathit{=}\beta_{0}+\mathbf{\beta}_{1}^{\top
}\mathbf{x}_{t}+u_{t},\ t\mathbf{\mathit{\in}}\mathbb{N}$.
\end{tabular}
}$\\
$\underset{\qquad}{\left.
\begin{tabular}
[c]{lll}%
\lbrack1] & Normality: & $\left(  y_{t}\mathbf{\mid X}_{t}\mathit{=}%
\mathbf{x}_{t}\right)  \backsim\mathsf{N}(.,.),$\\
\lbrack2] & Linearity: & $E\left(  y_{t}\mathbf{\mid X}_{t}\mathit{=}%
\mathbf{x}_{t}\right)  \mathit{=}\beta_{0}+\mathbf{\beta}_{1}^{\top}%
\mathbf{x}_{t}, $\\
\lbrack3] & Homoskedasticity: & $Var\left(  y_{t}\mathbf{\mid X}_{t}%
\mathit{=}\mathbf{x}_{t}\right)  \mathit{=}\sigma^{2},$\\
\lbrack4] & Independence: & $\{\left(  y_{t}\mathbf{\mid X}_{t}\mathit{=}%
\mathbf{x}_{t}\right)  ,\ t\mathbf{\mathit{\in}}\mathbb{N}\}$ indep.
process$,$\\
\lbrack5] & t-invariance: & $\left(  \beta_{0},\mathbf{\beta}_{1},\sigma
^{2}\right)  $ are \textit{not} changing with $t,$\\
\multicolumn{3}{l}{$\ \ \ \ \ \ \beta_{0}\mathit{=}E(y_{t})\mathit{-}%
\mathbf{\beta}_{1}^{\top}E(\mathbf{X}_{t}),\ \mathbf{\beta}_{1}\mathit{=}%
[Cov(\mathbf{X}_{t})]^{-1}Cov(\mathbf{X}_{t},y_{t}),\ $}\\
& \multicolumn{2}{l}{$\sigma^{2}\mathit{=}Var(y_{t})\mathit{-}Cov(\mathbf{X}%
_{t},y_{t})^{\top}[Cov(\mathbf{X}_{t})]^{-1}Cov(\mathbf{X}_{t},y_{t})$}%
\end{tabular}
\right\}  \ t\mathbf{\mathit{\in}}\mathbb{N}.}$\\\hline
\end{tabular}
\]

\textbf{Statistical adequacy}. An estimated LR model is said to be
\textit{statistically adequate} when all assumptions [1]-[5] are valid for
data $\mathbf{Z}_{0}.$ In practice, statistical adequacy can be appraised
using comprehensive misspecification testing; see Spanos (1999, 2015). The
importance of establishing statistical adequacy stems from the fact that it
secures the statistical reliability of inference based on such a model. That
is, the inference propositions associated with the LR model, including the
optimal properties of the MLE estimators and the relevant error probabilities
of the t and F tests, are reliable in the sense that their actual sampling
distributions approximate closely the theoretical ones derived by invoking the
validity of assumptions [1]-[5].

\textbf{Unreliability of inference}. When any subset of the assumptions
[1]-[5] are invalid, the reliability of inference of such procedures is called
into question. Statistical misspecifications are likely to give rise to
inconsistent estimators as well as induce sizeable discrepancies between the
nominal (assumed) error probabilities and the actual ones in testing. For
instance, when any of the assumptions [2], [4]-[5] are invalid, the OLS
estimators of $\left(  \beta_{0},\beta_{1}\right)  $ are likely to be
inconsistent, and the nominal error probabilities associated with the
significance t-tests for the coefficients $\left(  \beta_{0},\beta_{1}\right)
$ are likely to have significant discrepancies from the actual error
probabilities; see Spanos and McGuirk (2001), Spanos (2010). Applying a $.05$
significance level t-test when the actual type I error is closer to $.8$ is
likely to give rise to unreliable inferences.

It is important to emphasize that for assumptions [4] and [5] to be testable,
one needs to select an ordering of interest for data $\mathbf{Z}_{0}$. In the
case of time-series data, the ordering of interest is invariably `time', which
is an interval scale variable. For cross-section data, however, there are
often several orderings of interest, depending on the individual unit being
observed, and the modeler needs to think about such potential orderings as
they relate to [4]-[5]. Potential orderings for cross-section can vary from
gender (nominal scale), to age (ratio scale), etc.\vspace*{-0.15in}

\subsection{Statistical vs. substantive misspecification\vspace*{-0.08in}}

Let us return to example 1, where the problem of association reversal can be
viewed in the context of comparing the regression coefficients of $x_{1t},$
$\alpha_{1}$ and $\beta_{1},$ in the context of two Linear Regression
models:$\vspace*{-0.08in}$%
\[%
\begin{tabular}
[c]{ll}%
Model 1: & $y_{t}\mathit{=}\beta_{0}+\beta_{1}x_{1t}+\beta_{2}x_{2t}%
\mathit{+}u_{t},\ (u_{t}\mathit{\mid}\mathbf{X}_{t}\mathit{=}\mathbf{x}%
_{t})\mathit{\backsim}\text{\textsf{NIID}}(0,\sigma_{u}^{2}),\ t\mathit{\in
}\mathbb{N},\medskip$\\
Model 2: & $y_{t}\mathit{=}\alpha_{0}\mathit{+}\alpha_{1}x_{1t}+\varepsilon
_{t},\ (\varepsilon_{t}\mathit{\mid}X_{1t}\mathit{=}x_{1t})\mathit{\backsim
}\text{\textsf{NIID}}(0,\sigma_{\varepsilon}^{2}),\ t\mathit{\in}\mathbb{N},$%
\end{tabular}
\vspace*{-0.08in}%
\]
where $\mathbf{x}_{t}\mathit{:=}(x_{1t},x_{2t})^{\top}.$ In the previous
section, it was argued that when both models are statistically adequate, it
could happen that the estimated coefficients $\alpha_{1}$ and $\beta_{1}$
differ in both sign and magnitude. There is, however, a sizeable literature on
`omitted variables' which would call model 2 misspecified when $\beta_{2} $
turns out to be statistically significant; see Greene (2011). In what sense is
model 2 misspecified if its assumptions [1]-[5] (table 1) are valid?
Similarly, the literature on causal modeling would test the significance of
the covariances $\sigma_{13}$ and $\sigma_{23}$ as they relate to the
regression coefficients, to decide whether $x_{2t}$ is a confounder; see Pearl
(2011). How does this relate to the statistical misspecification perspective?

A closer look at the literature suggests that statistical misspecification is
often conflated with substantive misspecification, using confusing and
confused claims, such as the OLS estimator of $\alpha_{1}$ in model 2 is an
inconsistent estimator of $\beta_{1}$ in model 1 (Greene, 2011), ignoring the
fact that the two coefficients represent very different
parameterizations:$\ \alpha_{1}\mathit{=}\frac{\sigma_{12}}{\sigma_{22}%
},\ \beta_{1}\mathit{=}\frac{(\sigma_{12}\sigma_{33}-\sigma_{13}\sigma_{23}%
)}{(\sigma_{22}\sigma_{33}-\sigma_{23}^{2})}.$

To make any sense of such comparisons, one needs to distinguish between
\textit{statistical} and \textit{substantive adequacy} because the former
requires only that assumptions [1]-[5] are valid for $\mathbf{Z}_{0}$.
Assumptions [1]-[5] have nothing to do with: the LR model includes all
`substantively' relevant variables. The latter is a substantive assumption
that pertains to the explanatory potential of the estimated model as it
relates to the phenomenon of interest. Substantive inadequacy can arise from
missing but relevant variables, false causal claims, etc. The crucial
importance of this distinction stems from the fact that when models 1-2 are
statistically misspecified, both the test for an omitted variable, as well as
the tests for deciding whether $x_{2t}$ is a confounder, or a mediator, are
likely to give rise to untrustworthy results; see Spanos (2006b).

This distinction is also important when the term `spurious' is employed
without being qualified to differentiate between \textit{statistically }and
\textit{substantively spurious} inference results. Indeed, the term `spurious
correlation' is often used to describe the case where the statistical
significance of a correlation coefficient is taken at face value, and an
attempt is made to explain it away using substantive arguments; see Sober
(2001). More often than not, however, one can show that the statistical
significance is more apparent than real, because it is just an untrustworthy
result stemming from a statistically misspecified model.\vspace*{-0.15in}

\subsection{Example 2. Yule's `nonsense-correlations'\vspace*{-0.08in}}

The problem of `spurious' associations, first noted by Pearson (1896), was
high up in Yule's agenda during the first quarter of the 20th century,
returning to it on several occasions; see Yule (1909, 1910, 1921). Yule (1926)
is the culmination of his efforts to unravel the puzzle of `spurious' results
using the high correlations between time series data as an example. He used
data measuring the ratio of Church of England marriages to all marriages
($x_{t}$) and the mortality rate ($y_{t}$) over the period 1866-1911, to
demonstrate that their estimated correlation $\widehat{\rho}_{xy}%
\mathit{=}.9512$ was both very high and statistically significant. He
described this result as `nonsense-correlation' because `\textsf{common sense
judges to be incorrect}.\textsf{' (p.4)} He went on to reject any attempt,
however ingenious, to rationalize such a statistical result on substantive grounds:

\textsf{\textquotedblleft Now I suppose it is possible, given a little
ingenuity and goodwill, to rationalize very nearly anything. And I can imagine
some enthusiast arguing that the fall in the proportion of Church of England
marriages is simply due to the Spread of Scientific Thinking since 1866, and
the fall in mortality is also clearly to be ascribed to the Progress of
Science; hence both variables are largely or mainly influenced by a common
factor and consequently ought to be highly correlated. But most people would,
I think, agree with me that the correlation is simply sheer nonsense; that it
has no meaning whatever; that it is absurd to suppose that the two variables
in question are in any sort of way, however indirect, causally related to one
another.\textquotedblright\ (p. 2)}

Yule (1926) attempted to articulate the premise that `nonsense-correlations'
have something to do with the fact that his time series data are \textit{not}
`random series'. He could not establish a clear and direct link between
`spurious' associations and statistical misspecification, however, because he
was missing two key components that were yet to be integrated into statistics.
The first is the notion of a `parametric statistical model', innovated by
Fisher (1922), and the second is the theory of `stochastic processes' founded
by Kolmogorov (1933). The former comprises all the probabilistic assumptions
imposed on the data, and the latter formalizes the notions of a `random
series' into a realizationn of an IID stochastic proceses,\ as well as
departures from it in the form of probabilistic concepts for dependence and heterogeneity.

\textbf{Yule's reverse engineering}. Given that there was no notion of a
prespecified parametric statistical model, comprising the probabilistic
assumptions imposed on the data, Yule resorted to `reverse engineering':

\textsf{\textquotedblleft When we find that a theoretical formula applied to a
particular case gives results which common sense judges to be incorrect, it is
generally as well to examine the particular assumptions from which it was
deduced, and see which of them are inapplicable to the case in
point.\textquotedblright\ (p. 4-5)}

He went on to consider the formula for estimating the sample standard error
and elicit the implicit probabilistic assumptions that render it a `good'
estimator of the distribution standard error. Let us emulate Yule's reverse
engineering using the sample correlation coefficient, which is the focus of
his paper:$\vspace*{-0.08in}$%
\begin{equation}%
\begin{array}
[c]{c}%
\widehat{Corr(X_{t},Y_{t})}\mathit{=}\frac{%
\begin{array}
[c]{c}%
\frac{1}{n}\sum\nolimits_{t=1}^{n}(Y_{t}-\overline{Y})(X_{t}-\overline{X})
\end{array}
\smallskip}{\sqrt{\left[  \frac{1}{n}\sum\nolimits_{t=1}^{n}(X_{t}%
-\overline{X})^{2}\right]  \left[  \frac{1}{n}\sum\nolimits_{t=1}^{n}%
(Y_{t}-\overline{Y})^{2}\right]  }},\medskip\\%
\begin{array}
[c]{c}%
\overline{X}\mathit{=}\frac{1}{n}\sum\nolimits_{t=1}^{n}X_{t},\ \overline
{Y}\mathit{=}\frac{1}{n}\sum\nolimits_{t=1}^{n}Y_{t},\ \widehat{Var(X_{t}%
)}\mathit{=}\frac{1}{n}\sum\nolimits_{t=1}^{n}(X_{t}-\overline{X}%
)^{2},\medskip\\
\widehat{Var(Y_{t})}\mathit{=}\frac{1}{n}\sum\nolimits_{t=1}^{n}%
(Y_{t}\mathit{-}\overline{Y})^{2},\ \widehat{Cov(X_{t},Y_{t})}\mathit{=}%
\frac{1}{n}\sum\nolimits_{t=1}^{n}(Y_{t}\mathit{-}\overline{Y})(X_{t}%
\mathit{-}\overline{X}),
\end{array}
\end{array}
\label{cor}%
\end{equation}
as a `good' estimator of the distribution correlation coefficient:$\vspace
*{-0.08in}$%
\[%
\begin{array}
[c]{c}%
Corr(X_{t},Y_{t})\mathit{=}\frac{Cov(X_{t},Y_{t})}{\sqrt{Var(X_{t})Var(Y_{t}%
)}}%
\end{array}
\vspace*{-0.08in}%
\]
The first assumption implicit in these formulae is the \textit{constancy} of
the moments:\vspace*{-0.08in}%
\[%
\begin{array}
[c]{c}%
E(Y_{t})\mathit{=}\mu_{1},\ E(X_{t})\mathit{=}\mu_{2},\ Var(Y_{t}%
)\mathit{=}\sigma_{11},\ Var(X_{t})\mathit{=}\sigma_{22},\ Cov(X_{t}%
,Y_{t})\mathit{=}\sigma_{12},\ t\mathit{\in}\mathbb{N},
\end{array}
\vspace*{-0.08in}%
\]
which corresponds to a form of the \textit{ID assumption}. The formulae for
$\widehat{Var(X_{t})}$ and $\widehat{Var(Y_{t})},$ implicitly assume
\textit{non-correlation} over $t\mathit{\in}\mathbb{N}$, otherwise they should
have included covariances over $t\mathit{\in}\mathbb{N}$ terms. Yule also
sought to unveil the implicit distributional assumption
\textsf{\textquotedblleft in order to reduce the formula to the very simple
form given.\textquotedblright\ (p. 5)} The sample moments are not always
`optimal' estimators of the distribution moments. For instance, the estimators
in (\ref{cor}) will be `optimal' under Normality, but they will be non-optimal
if the distribution is Uniform; see Carlton (1946).

In light of the fact that under Normality the assumption of ID reduces to the
constancy of the first two moments, and non-correlation coincides with
\textit{Independence}, one could make a case that the implicit parametric
statistical model underlying the above formulae is the\textbf{\ simple
bivariate Normal} in table 2.%
\[%
\begin{tabular}
[c]{l}\hline%
\begin{tabular}
[c]{l}%
\textbf{Table 2 - The simple} (bivariate)\textbf{\ Normal model}%
\end{tabular}
\medskip\\\hline\hline
$\left.
\begin{tabular}
[c]{lll}%
\multicolumn{3}{l}{Statistical GM:\qquad\qquad$\mathbf{Z}_{t}=\mathbf{\mu
}+\mathbf{u}_{t},$}\\
\lbrack1] & Normal: & $\mathbf{Z}_{t}\backsim\mathsf{N}(.,.),$\\
\lbrack2] & Constant mean: & $E(\mathbf{Z}_{t})\mathit{=}\mathbf{\mu},$\\
\lbrack3] & Constant covariance: & $Var(\mathbf{Z}_{t})\mathit{=}%
\mathbf{\Sigma},$\\
\lbrack4] & Independence: & $\{\mathbf{Z}_{t},\ t$\textbf{$\in$}$\mathbb{N}\}$
is independent.\\\hline
\end{tabular}
\right\}  \;t$\textbf{$\in$}$\mathbb{N},$%
\end{tabular}
\]%
\begin{equation}%
\begin{array}
[c]{c}%
\mathbf{Z}_{t}\mathit{:=}\left(
\begin{array}
[c]{c}%
y_{t}\\
X_{t}%
\end{array}
\right)  ,\ \mathbf{\mu}\mathit{:=}\left(
\begin{array}
[c]{c}%
\mu_{1}\\
\mu_{2}%
\end{array}
\right)  ,\ \mathbf{\Sigma}\mathit{:=}\left(
\begin{array}
[c]{cc}%
\sigma_{11} & \sigma_{12}\\
\sigma_{12} & \sigma_{22}%
\end{array}
\right)
\end{array}
\label{bn}%
\end{equation}

When any of the assumptions [1]-[4] are invalid for the particular data
$\mathbf{Z}_{0}$, the estimated correlation coefficient is likely to be
`spurious' (statistically untrustworthy). Granted, certain departures from
particular assumptions, such as [2]-[4], are more serious than other
departures, say from [1]. A glance at the t-plots of Yule's (1926) data
suggests, to borrow his phrase on p. 5, that:

\textsf{\textquotedblleft Neither series, obviously, in the least resembles a
random series\textquotedblright}\ (aka IID).%
\[%
{\parbox[b]{2.7959in}{\begin{center}
\includegraphics[
natheight=3.464400in,
natwidth=5.190600in,
height=1.7331in,
width=2.7959in
]%
{../../../swp55/Docs-toshiba/O6M4V000.wmf}%
\\%
\protect\begin{tabular}
[c]{l}%
Fig. 1: t-plot of $x_{t}$-ratio of Church of\protect\\
England marriages to all marriages
\protect\end{tabular}
\end{center}}}
{\parbox[b]{2.6922in}{\begin{center}
\includegraphics[
natheight=3.464400in,
natwidth=5.190600in,
height=1.7331in,
width=2.6922in
]%
{../../../swp55/Docs-toshiba/O6M4V001.wmf}%
\\%
\protect\begin{tabular}
[c]{l}%
Fig. 2: t-plot of $y_{t}$-the mortality rate\protect\\
for the period 1866-1911
\protect\end{tabular}
\end{center}}}
\]
Both data series exhibit clear departures from IID (fig. 6) in the form of
mean $t$-heterogeneity (trending mean) and dependence (irregular cycles). To
bring out the cycles in the original data more clearly one needs to subtract
the trending means using, say, a generic 3rd degree trend polynomial.%
\[%
{\parbox[b]{2.6714in}{\begin{center}
\includegraphics[
natheight=3.464400in,
natwidth=5.190600in,
height=1.7331in,
width=2.6714in
]%
{../../../swp55/Docs-toshiba/O6OF6X00.wmf}%
\\
Fig. 3: t-plot of detrended $x_t$
\end{center}}}
{\parbox[b]{2.629in}{\begin{center}
\includegraphics[
natheight=3.464400in,
natwidth=5.190600in,
height=1.7331in,
width=2.629in
]%
{../../../swp55/Docs-toshiba/O6OF6X01.wmf}%
\\
Fig. 4: t-plot of detrended $y_t$
\end{center}}}
\]

In light of the direct relationship between the correlation ($\rho_{12}$) and
the regression coefficient ($\beta_{1}$) in (\ref{cr}), one can pose the
question of statistical adequacy in the context of the Linear Regression
model, which will yield:\vspace*{-0.07in}%
\begin{equation}%
\begin{array}
[c]{c}%
y_{t}\mathbf{\mathit{=}}\underset{(1.416)}{-10.847}\mathit{+}%
\underset{(.020)}{.419}x_{t}+\widehat{u}_{t},\;R^{2}\mathbf{\mathit{=}%
}.905,\;s\mathbf{\mathit{=}}.664,\;n\mathbf{\mathit{=}}46,
\end{array}
\vspace*{-0.1in}\label{eq5}%
\end{equation}
where the standard errors are reported in brackets below the coefficient
estimates. Both coefficients $\left(  \beta_{0},\beta_{1}\right)  $ seem
statistically significant since the t-ratios are:\vspace*{-0.1in}%
\[%
\begin{array}
[c]{cc}%
\tau_{0}(\mathbf{z}_{0})\mathit{=}\frac{10.847}{1.416}\mathit{=}7.660[.000], &
\tau_{1}(\mathbf{z}_{0})\mathit{=}\frac{.419}{.020}\mathit{=}20.95[.000],
\end{array}
\vspace*{-0.1in}%
\]
and the p-values are given in square brackets. Note that the implied
correlation (see (\ref{cor})) yields the value in Yule (1926): $\widehat{\rho
}_{xy}\mathit{=}.419(\frac{4.854}{2.137})\mathit{=}.952[.000]$.

A glance at the t-plot of the residuals (fig. 5), however, indicates that
(\ref{eq5}) is statistically misspecified; assumptions [4]-[5] are likely to
be invalid. The residual t-plot differs from that of a NIID realization (fig.
6) in so far as it exhibits distinct trends and cycles. These
misspecifications are confirmed formally by the statistical significance of
the trends and lags in the auxiliary regression based on the residuals
($\widehat{u}_{t}$) from (\ref{eq5}):\vspace*{-0.18in}%
\begin{equation}%
\begin{array}
[c]{c}%
\widehat{u}_{t}\mathbf{\mathit{=}}\underset{(1.267)}{11.987}\mathit{-}%
\underset{(.016)}{.413}x_{t}\mathit{-}\underset{(.998)}{1.670}t\mathit{-}%
\underset{(.216)}{.406}t^{2}\mathit{+}\underset{(.076)}{.885}y_{t-1}%
\mathit{+}\underset{(.015)}{.006}x_{t-1}%
\end{array}
\label{eq6}%
\end{equation}%
\[%
{\parbox[b]{2.7821in}{\begin{center}
\includegraphics[
natheight=3.464400in,
natwidth=5.190600in,
height=1.7331in,
width=2.7821in
]%
{../../../swp55/Docs-toshiba/O6KGIS02.wmf}%
\\
Fig. 5: t-plot of the residuals from (\ref{eq5})
\end{center}}}
{\parbox[b]{2.6221in}{\begin{center}
\includegraphics[
natheight=3.464400in,
natwidth=5.190600in,
height=1.7331in,
width=2.6221in
]%
{../../../swp55/Docs-toshiba/O6KGIS03.wmf}%
\\
Fig. 6: t-plot of NIID data
\end{center}}}
\]
These results suggest that the estimator of $\beta_{1}$ is
inconsistent\textit{,} and the t-test for its significance is statistically
untrustworthy. Taking mean deviations from $\left(  \overline{x},\overline
{y}\right)  $ when the actual means are trending, will render all the above
estimators in (\ref{cor}) inconsistent.

In light of these departures from the IID assumptions,\ (\ref{cor}) is an
\textit{inconsistent} estimator of the correlation coefficient, and thus
statistically spurious. Indeed, one can easily show that when the data are
de-trended and de-memorized (subtract the temporal dependence using 2 lags) to
render\ (\ref{cor}) an appropriate estimator, the estimated correlation is:
$\widehat{\rho}_{xy}\mathit{=}.003[.985],$ which is totally statistically insignificant.

In summary, the notion of statistical adequacy provides a direct and testable
link between statistical misspecification and statistically untrustworthy
(spurious) associations, or inference results more generally. A likely
criticism of this link is that the probability assumptions of the assumed
model in (\ref{bn}) are too strong, in contrast to the current statistical
practice favoring as weak a set of assumptions as possible. The short reply to
such a charge is that weaker but non-testable assumptions (i) do not render
the assumed model less vulnerable to statistical misspecifications, and (ii)
they underestimate the importance of securing statistical adequacy. In
addition, weak assumptions often rely on asymptotic sampling distributions
without testing the validity of the assumptions invoked by limit theorems; see
Spanos (2015). The truth of the matter is that the trustworthiness of all
inference results will rely exclusively on the approximate validity of the
probabilistic assumptions imposed on $\mathbf{Z}_{0},$ and nothing else. As
argued by Le Cam (1986, p. xiv): \textsf{\textquotedblleft... limit theorems
"as }$n$\textsf{\ tends to infinity"\ are logically devoid of content about
what happens at any particular }$n$\textsf{. }\vspace*{-0.15in}

\subsection{\textbf{Example 3}. \textbf{The third variable reversion}%
\vspace*{-0.1in}}

In this sub-section we consider an empirical example based on cross-section
data because statistical adequacy is less well appreciated in such a context.

Consider the case where a practitioner wants to evaluate the effect of
education on a person's income. The data refer to education, $x_{t}$-years of
schooling, and income, $y_{t}$-thousands of dollars, for $n\mathit{=}100$
working people within the age group of 30-40 years old selected from a city's
population. The estimated LR model yields:\vspace*{-0.08in}%
\begin{equation}%
\begin{array}
[c]{c}%
y_{t}\mathbf{\mathit{=}}\underset{(1.957)}{53.694}\mathit{-}%
\underset{(.147)}{.474}x_{t}+\widehat{u}_{t},\;R^{2}\mathbf{\mathit{=}%
}.096,\;s\mathbf{\mathit{=}}3.307,\;n\mathbf{\mathit{=}}100.
\end{array}
\vspace*{-0.08in}\label{eq1}%
\end{equation}
Both coefficients $\left(  \beta_{0},\beta_{1}\right)  $ appear to be
statistically significant since the t-ratios are:\vspace*{-0.08in}%
\[%
\begin{array}
[c]{cc}%
\tau_{0}(\mathbf{z}_{0})\mathit{=}\frac{53.694}{1.957}\mathit{=}%
27.437[.0000], & \tau_{1}(\mathbf{z}_{0})\mathit{=}\frac{.474}{.147}%
\mathit{=}3.224[.001].
\end{array}
\vspace*{-0.08in}%
\]
The practitioner is surprised by the negative sign of the coefficient of
$x_{t},$ since that implies that additional years of education contribute
negatively to one's income. He takes a closer look at the data and decides to
run separate linear regressions for men ($n_{1}\mathbf{\mathit{=}}50$) and
women ($n_{2}\mathbf{\mathit{=}}50$).

The estimated LR model for men yields:\vspace*{-0.08in}
\begin{equation}%
\begin{array}
[c]{c}%
y_{1t}\mathbf{\mathit{=}}\underset{(2.236)}{45.229}\mathit{+}%
\underset{(.172)}{.409}x_{1t}+\widehat{u}_{1t},\;R^{2}\mathbf{\mathit{=}%
}.973,\;s\mathbf{\mathit{=}}2.371,\;n_{1}\mathbf{\mathit{=}}50.
\end{array}
\vspace*{-0.1in}\label{eq2}%
\end{equation}

The estimated LR model for women yields:\vspace*{-0.08in}%
\begin{equation}%
\begin{array}
[c]{c}%
y_{2t}\mathbf{\mathit{=}}\underset{(2.937)}{35.106}\mathit{+}%
\underset{(.199)}{.675}x_{2t}+\widehat{u}_{2t},\;R^{2}\mathbf{\mathit{=}%
}.193,\;s\mathbf{\mathit{=}}2.124,\;n_{2}\mathbf{\mathit{=}}50.
\end{array}
\vspace*{-0.1in}\label{eq3}%
\end{equation}

The estimation results in (\ref{eq2})-(\ref{eq3}) indicate that for both
estimated regressions:\smallskip

-- the coefficients $\left(  \beta_{0},\beta_{1}\right)  $ are statistically
significant, and\smallskip

-- the sign of the coefficient $\beta_{1},$ of education variable $(x_{t})$,
is positive.

The positive sign of the estimated $\beta_{1}$ clearly contradicts the
negative sign in (\ref{eq1}), which is usually interpreted as a case where a
statistical association is reversed. This is considered as an example of
Simpson's paradox when viewed from perspective (b), where gender ($D_{t}$) is
viewed as a confounding variable that correlates with both $y_{t}$-income and
$x_{t}$-education. In econometrics, this is usually viewed as a case of
`omitted-variable bias'; see Greene (2011). According to Pearl (2014), p. 10,
the only way to decide whether to rely on the aggregated data regression in
(\ref{eq1}) or the disaggregated data regressions (\ref{eq2})-(\ref{eq3}) is
to use causal calculus.

\newpage

\vspace*{-0.5in}%

\[%
{\parbox[b]{3.039in}{\begin{center}
\includegraphics[
natheight=3.464400in,
natwidth=5.190600in,
height=1.7331in,
width=3.039in
]%
{../../../swp55/Docs-toshiba/O6HX9702.wmf}%
\\
Fig. 7: Residuals from equation (\ref{eq1})
\end{center}}}
\]

Upon reflection, however, the statistical misspecification perspective
provides an alternative way to resolve the paradox on statistical adequacy
grounds. The above estimation and testing results in (\ref{eq1})-(\ref{eq3})
are trustworthy only when the model assumptions [1]-[5] are valid for the
particular data for each of the three estimated equations. Estimating the
aggregated data equation (\ref{eq1}) using `gender' as the ordering of
interest, and plotting the residuals (fig. 7) suggests that (\ref{eq1}) is
\textit{statistically misspecified} because the t-plot is far from being
Normal white-noise. Assumption [5] is clearly invalid since its sample mean is
not constant around zero, but shifts from positive for the first half to
negative for the second, and the variance appears smaller for the second half;
see Spanos (1999), ch. 5. This form of t-heterogeneity differs from that in
Yule's data discussed above.

This is confirmed by the auxiliary regression using the residuals
($\widehat{u}_{t}$): \vspace*{-0.08in}%
\[%
\begin{array}
[c]{c}%
\widehat{u}_{t}\mathbf{\mathit{=}}\underset{(2.00)}{-15.986}\mathit{+}%
\underset{(.134)}{.967}x_{t}+\underset{(.616)}{6.556}D_{t},\;R^{2}%
\mathbf{\mathit{=}}.54,\;s\mathbf{\mathit{=}}2.263,\;n\mathbf{\mathit{=}}100,
\end{array}
\vspace*{-0.12in}%
\]
where $D_{t}\mathit{:=}(1,1,...,1,0,0...,0)$, 1-male, 0-female, since its
coefficient is statistically significant: $\tau_{3}(\mathbf{z}_{0}%
)\mathit{=}\frac{6.556}{.616}\mathit{=}10.643[.0000].$

This suggests that the statistical misspecification perspective provides a
very different interpretation of the reversion results and offers an
alternative way to resolve the apparent paradox.

First, the key to resolving any seemingly conflicting inference results is not
the notion of `confounding' (Pearl, 2014), but that of statistical adequacy.
Before one can talk about any form of reversal of a statistical association,
one needs to establish that all the associations involved are statistically
trustworthy. Any claim that there is a `reversal of association' between
equation (\ref{eq1}) and (\ref{eq2})-(\ref{eq3}) is misleading since the
aggregated data equation (\ref{eq1}) is statistically misspecified. Therefore,
the inference that the coefficient of $x_{t}$ is negative and statistically
significant is untrustworthy; an artifact of imposing invalid probabilistic
assumptions on data $\mathbf{z}_{0}$. Hence, the aggregate data misrepresent
the relationship between $y_{t}$ and $x_{t}.$

\newpage

\vspace*{-0.5in}%

\[%
{\parbox[b]{2.6221in}{\begin{center}
\includegraphics[
natheight=3.464400in,
natwidth=5.190600in,
height=1.7331in,
width=2.6221in
]%
{../../../swp55/Docs-toshiba/O6PLIX02.wmf}%
\\
Fig. 8: t-plot of income ($y_t)$
\end{center}}}
{\parbox[b]{2.6221in}{\begin{center}
\includegraphics[
natheight=3.464400in,
natwidth=5.190600in,
height=1.7331in,
width=2.6221in
]%
{../../../swp55/Docs-toshiba/O6PLIX03.wmf}%
\\
Fig. 9: t-plot of education ($x_t)$
\end{center}}}
\]

Second, the diagnosis that the variable `gender', represented by $D_{t}$ is a
missing `confounder' seems rather misleading for two reasons. The information
pertaining to the ordering(s) of potential interest is already in the original
data $\mathbf{z}_{0}$ (see figures 8-9). In addition, defining the confounder
as an omitted variable $Z$ which is related to the included variables $X$ and
$Y$ in the right way, requires that $Z$ is stochastic variable, not a
deterministic ordering. For statistical inference purposes, the inclusion of
generic terms such as shifts in the mean, trends and lags in the estimated
equation could, in certain cases, secure statistical adequacy, without having
to resort to finding additional explanatory variables.

In the case of (\ref{eq1}), a more pertinent explanation is that the modeler
neglected, or chose to ignore, the heterogeneity in the data by assuming
constant mean and variance (ID) for both data series with respect to the
ordering, \textit{gender}. Such forms of misspecification pertain to
statistical information contained in the data, which could be generically
modeled using shift functions or/and trend polynomials in $t$ or/and lags,
respectively; see Spanos (1999).

In the case of example 3, one could attempt to respecify the original equation
in (\ref{eq1}) by including the dummy variable ($D_{t}$):\vspace*{-0.08in}%
\begin{equation}%
\begin{array}
[c]{c}%
y_{t}\mathbf{\mathit{=}}\underset{(2.00)}{37.639+}\underset{(.616)}{6.556}%
D_{t}+\underset{(.134)}{.501}x_{t}+\widehat{u}_{t},\;R^{2}\mathbf{\mathit{=}%
}.58,\;s\mathbf{\mathit{=}}2.272,\;n\mathbf{\mathit{=}}100.
\end{array}
\vspace*{-0.08in}\label{eq4}%
\end{equation}
As it stands, the coefficient of $x_{t}$ represents an misleading weighted
average of the two coefficients from the disaggregated data in (\ref{eq2}%
)-(\ref{eq3}). The residuals from this equation (fig. 10) do not indicate any
major departures from assumptions [1]-[5], but in practice one needs to apply
thorough misspecification testing to confirm or deny such a claim. For
instance, one needs to test that the variances of the residuals in the two
sub-samples are equal; see Spanos (1986), p. 481-3.

Finally and most importantly, using the statistical misspecification
perspective, one can distinguish clearly between example 1 (section 2), and
examples 2 and 3 above. The key difference is that in example 1 both LR models
(\ref{rm1}) and (\ref{rm2}) are statistically adequate. In contrast, in
example 3 the estimated LR model (\ref{eq1}) based on aggregated data, is
statistically misspecified which renders the estimated coefficients and
t-tests statistically untrustworthy. Hence, there was never a statistically
trustworthy result at the aggregate level that gave rise to a reversal of
associations. In example 3, only the disaggregated data give rise to
statistically reliable inferences. This calls seriously into question the
conventional wisdom that these two cases as identical, as stated by Samuels
(1993), p. 87:

\textsf{\textquotedblleft Simpson's paradox is actually no more paradoxical
than the reversal or distortion of association in other settings, no more, for
instance, than the familiar fact that a partial regression coefficient can
have a different sign from a simple regression coefficient.\textquotedblright}%

\[%
{\parbox[b]{3.039in}{\begin{center}
\includegraphics[
natheight=3.464400in,
natwidth=5.190600in,
height=1.7331in,
width=3.039in
]%
{../../../swp55/Docs-toshiba/O6LWNA00.wmf}%
\\
Fig. 10: Residuals from equation (\ref{eq4})
\end{center}}}
\]

In concluding this section, it is important to emphasize that the statistical
misspecification perspective requires one to know the complete set of
probabilistic assumptions imposed on the data, i.e. the statistical model.
More often than not, practitioners have an incomplete picture of the
statistical model, they rarely test its assumptions, and thus the ensuing
inference results are often untrustworthy. Hence, in evaluating published
empirical papers, it is sometimes useful to employ Yule's reverse engineering
to uncover the statistical model. \vspace*{-0.15in}

\section{Misspecified Bernoulli models\vspace*{-0.1in}}

In this section we will revisit two cross-section data sets that have been
widely discussed in the statistics literature, using the statistical
misspecification perspective.\vspace*{-0.15in}

\subsection{\textbf{Example 4. The UC Berkeley admissions data}\vspace
*{-0.1in}}

Bickel et al. (1975) published an influential paper in \textit{Science}, where
they illustrated Simpson's paradox using cross-section data based on UC
Berkeley admissions, for the Fall of 1973. Their perspective relates to
perspective (a) and pertains to the reversal of a statistical relationship
between the aggregated data, at the university level, and the disaggregated
data, at the department level. The aggregate data are shown in table 3 and the
data for the largest 5 departments, denoted by A-F, are given in table 4;

see \textrm{https://en.wikipedia.org/wiki/Simpson\%27s\_paradox}.

The estimated parameter $\theta\mathit{=}\mathbb{P}(X\mathit{=}1)$ based on
the aggregate data (table 3) indicates that the rate of admissions for female
candidates ($\ \widehat{\theta}_{F}\mathit{=}.35$) is smaller that for male
candidates ($\widehat{\theta}_{M}\mathit{=}.44$), and a test for the
difference indicated a statistically significant difference; see Bickel et al.
(1975). At the department level, however, the admissions rate for females is
greater than that of males in five out of the six departments shown in table
4. This is interpreted as an apparent reversal of the inference based on the
aggregate data. The statistical misspecification perspective, however,
suggests that the estimated admissions rate using the aggregated\ data is
statistically untrustworthy. Let us unpack this claim.%
\[%
\begin{tabular}
[c]{c}\hline%
\begin{tabular}
[c]{l}%
\textbf{Table 3:} Admissions Aggregate Data
\end{tabular}
\medskip\\\hline\hline%
\begin{tabular}
[c]{|l||ll||r|}\hline
& \textbf{M}$\text{\textbf{ales}}$ & \textbf{F}$\text{\textbf{emales}}$ &
$\text{\textbf{total}}$\\\hline\hline
Admit & $3738$ & $1494$ & $5232$\\\hline
Deny & $4704$ & $2827$ & $7531$\\\hline\hline
Total & $8442$ & $4321$ & $12763$\\\hline\hline
\multicolumn{4}{|l|}{$%
\begin{array}
[c]{c}%
\widehat{\theta}_{M}\mathit{=}\frac{3738}{8442}\mathit{=}.44,\ \widehat{\theta
}_{F}\mathit{=}\frac{1494}{4321}\mathit{=}.35
\end{array}
$}\\\hline
\end{tabular}
\end{tabular}
\]

\[%
\begin{tabular}
[c]{ll}\hline
\multicolumn{2}{c}{%
\begin{tabular}
[c]{l}%
\textbf{Table 4: Admissions disaggregated data for departments A-F}%
\end{tabular}
\medskip}\\\hline\hline
$%
\begin{tabular}
[c]{|l||ll||r|}\hline
\fbox{%
\begin{tabular}
[c]{l}%
A
\end{tabular}
} & \textbf{M}$\text{\textbf{ales}}$ & \textbf{F}$\text{\textbf{emales}}$ &
$\text{\textbf{total}}$\\\hline\hline
Admit & $512$ & $89$ & $601$\\\hline
Deny & $313$ & $19$ & $332$\\\hline\hline
Total & $825$ & $108$ & 933\\\hline\hline
\multicolumn{4}{|l|}{$%
\begin{array}
[c]{c}%
\widehat{\theta}_{AM}\mathit{=}\frac{512}{825}\mathit{=}.62,\ \widehat{\theta
}_{AF}\mathit{=}\frac{89}{108}\mathit{=}.82
\end{array}
$}\\\hline
\end{tabular}
$ & $%
\begin{tabular}
[c]{|l||ll||r|}\hline
\fbox{%
\begin{tabular}
[c]{l}%
B
\end{tabular}
} & \textbf{M}$\text{\textbf{ales}}$ & \textbf{F}$\text{\textbf{emales}}$ &
$\text{\textbf{total}}$\\\hline\hline
Admit & $353$ & $17$ & $370$\\\hline
Deny & $207$ & $8$ & $215$\\\hline\hline
Total & $560$ & $25$ & $585$\\\hline\hline
\multicolumn{4}{|l|}{$%
\begin{array}
[c]{c}%
\widehat{\theta}_{BM}\mathit{=}\frac{353}{560}\mathit{=}.63,\ \widehat{\theta
}_{BF}\mathit{=}\frac{17}{25}\mathit{=}.68
\end{array}
$}\\\hline
\end{tabular}
$\\
\multicolumn{2}{l}{}\\
$%
\begin{tabular}
[c]{|l||ll||r|}\hline
\fbox{%
\begin{tabular}
[c]{l}%
C
\end{tabular}
} & \textbf{M}$\text{\textbf{ales}}$ & \textbf{F}$\text{\textbf{emales}}$ &
$\text{\textbf{total}}$\\\hline\hline
Admit & $120$ & $202$ & $322$\\\hline
Deny & $205$ & $391$ & $596$\\\hline\hline
Total & $325$ & $593$ & $918$\\\hline\hline
\multicolumn{4}{|l|}{$%
\begin{array}
[c]{c}%
\widehat{\theta}_{CM}\mathit{=}\frac{120}{325}\mathit{=}.37,\ \widehat{\theta
}_{CF}\mathit{=}\frac{202}{593}\mathit{=}.34
\end{array}
$}\\\hline
\end{tabular}
$ & $%
\begin{tabular}
[c]{|l||ll||r|}\hline
\fbox{%
\begin{tabular}
[c]{l}%
D
\end{tabular}
} & \textbf{M}$\text{\textbf{ales}}$ & \textbf{F}$\text{\textbf{emales}}$ &
$\text{\textbf{total}}$\\\hline\hline
Admit & $139$ & $131$ & $270$\\\hline
Deny & $278$ & $244$ & $522$\\\hline\hline
Total & $417$ & $375$ & $792$\\\hline\hline
\multicolumn{4}{|l|}{$%
\begin{array}
[c]{c}%
\widehat{\theta}_{DM}\mathit{=}\frac{139}{417}\mathit{=}.33,\ \widehat{\theta
}_{DF}\mathit{=}\frac{131}{375}\mathit{=}.35
\end{array}
$}\\\hline
\end{tabular}
$\\
\multicolumn{2}{l}{}\\\cline{2-2}%
$%
\begin{tabular}
[c]{|l||ll||r|}\hline
\fbox{%
\begin{tabular}
[c]{l}%
E
\end{tabular}
} & \textbf{M}$\text{\textbf{ales}}$ & \textbf{F}$\text{\textbf{emales}}$ &
$\text{\textbf{total}}$\\\hline\hline
Admit & $53$ & $94$ & $147$\\\hline
Deny & $138$ & $199$ & $337$\\\hline\hline
Total & $191$ & $293$ & $484$\\\hline\hline
\multicolumn{4}{|l|}{$%
\begin{array}
[c]{c}%
\widehat{\theta}_{EM}\mathit{=}\frac{53}{191}\mathit{=}.28,\ \widehat{\theta
}_{EF}\mathit{=}\frac{94}{293}\mathit{=}.32
\end{array}
$}\\\hline
\end{tabular}
$ & \multicolumn{1}{l|}{$%
\begin{tabular}
[c]{|l||ll||r|}\hline
\fbox{%
\begin{tabular}
[c]{l}%
F
\end{tabular}
} & \textbf{M}$\text{\textbf{ales}}$ & \textbf{F}$\text{\textbf{emales}}$ &
$\text{\textbf{total}}$\\\hline\hline
Admit & $22$ & $23$ & $45$\\\hline
Deny & $351$ & $318$ & $669$\\\hline\hline
Total & $373$ & $341$ & $714$\\\hline\hline
\multicolumn{4}{|l|}{$%
\begin{array}
[c]{c}%
\widehat{\theta}_{FM}\mathit{=}\frac{22}{373}\mathit{=}.06,\ \widehat{\theta
}_{FF}\mathit{=}\frac{23}{341}\mathit{=}.07
\end{array}
$}\\\hline
\end{tabular}
$}%
\end{tabular}
\]

What is missing from the discussion of the traditional association reversal
interpretation is any evidence that the above inferences based on the
estimated $\theta$ is trustworthy. Such evidence can be secured by testing the
validity of the assumptions invoked by the above inferences, which comprise
the underlying statistical model: a \textit{bivariate} version of the simple
Bernoulli model (table 5), with $\theta$ replaced with a vector of unknown
parameters $\mathbf{\theta}\mathit{:=}(\theta_{00},\theta_{01},\theta
_{10},\theta_{11})$; see Bishop et al. (1975).

In relation to the Bernoulli model, it is important to point out that $\theta$
is also the mean of process underlying the data, as well as determining the
variance, i.e.$\vspace*{-0.07in}$%
\[%
\begin{array}
[c]{c}%
E(X_{t})\mathit{=}\theta,\ Var(X_{t})\mathit{=}\theta(1\mathit{-}%
\theta),\ 0\leq\theta\leq1,\ \forall t\mathbf{\mathit{\in}}\mathbb{N}.
\end{array}
\vspace*{-0.07in}%
\]

\[%
\begin{tabular}
[c]{l}\hline%
\begin{tabular}
[c]{l}%
\textbf{Table 5 - The simple Bernoulli model}%
\end{tabular}
\medskip\\\hline\hline
$\left.
\begin{tabular}
[c]{lll}%
\multicolumn{3}{l}{Statistical GM:\qquad\qquad$X_{t}=\theta+u_{t},\;$}\\
\lbrack1] & Bernoulli: & $X_{t}\backsim\mathsf{Ber}(.,.),$\\
\lbrack2] & Constant mean: & $E(X_{t})\mathit{=}\theta,$\\
\lbrack3] & Constant variance: & $Var(X_{t})\mathit{=}\theta(1-\theta),$\\
\lbrack4] & Independence: & $\{X_{t},\ t$\textbf{$\mathit{\in}$}$\mathbb{N}\}$
is independent.
\end{tabular}
\right\}  t$\textbf{$\mathit{\in}$}$\mathbb{N}$\\\hline
\end{tabular}
\]

A glance at the estimated $\theta$ for males and females at the department
level in table 4, indicate clearly that the estimated means and variances
differ, not only from those based on aggregate data, but also between
departments. This renders assumptions [2] and [3] invalid at the aggregate
level. That is, when the data are aggregated the process $\{X_{t}%
,\ t\mathit{\in}\mathbb{N}\}$ is no longer Identically Distributed (ID) with
respect to the ordering `gender'.

In light of this, the association reversal is spurious because the estimated
values:$\vspace*{-0.07in}$%
\[%
\begin{array}
[c]{c}%
\widehat{\theta}_{M}\mathit{=}\frac{3738}{8442}\mathit{=}.44,\ \widehat{\theta
}_{F}\mathit{=}\frac{1494}{4321}\mathit{=}.35
\end{array}
\vspace*{-0.07in}%
\]
from the aggregated data. This is because the estimators of $\theta$ using the
aggregated data will be \textit{inconsistent} estimators of the true $\theta$.%
\[
\hspace*{-0.25in}%
{\parbox[b]{2.6645in}{\begin{center}
\includegraphics[
natheight=3.464400in,
natwidth=5.190600in,
height=1.7331in,
width=2.6645in
]%
{../../../swp55/Docs-toshiba/O6LWRG03.wmf}%
\\%
\protect\begin{tabular}
[c]{l}%
Fig. 11: t-plot of $X_{t}\backsim$\textsf{BerIID}$(.2,.16)$%
\protect\end{tabular}
\end{center}}}
{\parbox[b]{2.7475in}{\begin{center}
\includegraphics[
natheight=3.464400in,
natwidth=5.190600in,
height=1.7331in,
width=2.7475in
]%
{../../../swp55/Docs-toshiba/O6LWRG04.wmf}%
\\%
\protect\begin{tabular}
[c]{l}%
Fig. 12: t-plot of $Z_{t}\backsim$\textsf{BerIID}$(.6,.24)$%
\protect\end{tabular}
\end{center}}}
\]

To see how this arises in practice, consider figures 11-12 that represent the
t-plots of two Bernoulli IID [\textsf{BerIID}($E(X_{t}),\ Var(X_{t}))]$
processes with $\theta\mathit{=}.2$ and $\theta\mathit{=}.6,$ respectively.

As can be seen from these figures, the concentration of longer `runs' [group
of successive values of 0's or 1's] switches from the value 0 to the value 1
as $\theta$ increases above $.5$. Hence, any attempt to ignore the differences
in the two moments of such processes will give rise to a misspecified
Bernoulli model. That invalidates any inferences based on the aggregate data,
and the only potentially reliable inference can be drawn from the
disaggregated data.\vspace*{-0.15in}

\subsection{Example 5. The Lindley and Novick (1981) data\vspace*{-0.1in}}

The above statistical misspecification perspective can be used to explain the
seemingly contradictory results in Lindley and Novick's (1981) hypothetical
data shown below. This is a particularly interesting example because, as
argued by Armistead (2014), the ordering of interest might become apparent
after the data are collected. For instance, in a clinical trial the `gender'
or/and `age' ordering(s) might turn out to be relevant after the data are collected.

The estimated $\theta$'s for the aggregated data in table 6:$\vspace
*{-0.07in}$%
\[%
\begin{array}
[c]{c}%
\widehat{\theta}_{W}\mathit{=}\frac{20}{40}\mathit{=}.5,\ \widehat{\theta}%
_{B}\mathit{=}\frac{16}{40}\mathit{=}.4,
\end{array}
\vspace*{-0.07in}%
\]
are very different from those based on the disaggregated data (table 7),
rendering the former statistically untrustworthy because it imposes an invalid
assumption: the means of the Bernoulli process underlying the disaggregated
data are constant.%
\[%
\begin{tabular}
[c]{l}\hline%
\begin{tabular}
[c]{l}%
\textbf{Table 6: Lindley-Novick}\\
\textbf{\ \ \ \ Aggregated Data}%
\end{tabular}
\medskip\\\hline\hline%
\begin{tabular}
[c]{|l||ll||r|}\hline
& \textbf{White} & \textbf{Black} & $\text{\textbf{total}}$\\\hline\hline
High & $20$ & $16$ & $36$\\\hline
Low & $20$ & $24$ & $44$\\\hline\hline
Total & $40$ & $40$ & $80$\\\hline
\multicolumn{4}{|l|}{$%
\begin{array}
[c]{c}%
\widehat{\theta}_{W}\mathit{=}\frac{20}{40}\mathit{=}.5,\ \widehat{\theta}%
_{B}\mathit{=}\frac{16}{40}\mathit{=}.4
\end{array}
$}\\\hline
\end{tabular}
\end{tabular}
\]%
\[%
\begin{tabular}
[c]{ll}\hline
\multicolumn{2}{c}{%
\begin{tabular}
[c]{l}%
\textbf{Table 7: Lindley-Novick disaggregated data}%
\end{tabular}
}\\\hline\hline%
\begin{tabular}
[c]{|l||ll||r|}\hline
Short & \textbf{White} & \textbf{Black} & $\text{\textbf{total}}%
$\\\hline\hline
High & \multicolumn{1}{||c}{$2$} & \multicolumn{1}{c||}{$9$} &
\multicolumn{1}{||c|}{$11$}\\\hline
Low & \multicolumn{1}{||c}{$8$} & \multicolumn{1}{c||}{$21$} &
\multicolumn{1}{||c|}{$29$}\\\hline\hline
Total & \multicolumn{1}{||c}{$10$} & \multicolumn{1}{c||}{$30$} &
\multicolumn{1}{||c|}{$40$}\\\hline\hline
\multicolumn{4}{|l|}{$%
\begin{array}
[c]{c}%
\widehat{\theta}_{SW}\mathit{=}\frac{2}{10}\mathit{=}.2,\ \widehat{\theta
}_{SB}\mathit{=}\frac{9}{30}\mathit{=}.3
\end{array}
$}\\\hline
\end{tabular}
&
\begin{tabular}
[c]{|l||ll||r|}\hline
Tall & \textbf{White} & \textbf{Black} & $\text{\textbf{total}}$\\\hline\hline
High & \multicolumn{1}{||c}{$18$} & \multicolumn{1}{c||}{$7$} &
\multicolumn{1}{||c|}{$25$}\\\hline
Low & \multicolumn{1}{||c}{$12$} & \multicolumn{1}{c||}{$3$} &
\multicolumn{1}{||c|}{$15$}\\\hline\hline
Total & \multicolumn{1}{||c}{$30$} & \multicolumn{1}{c||}{$10$} &
\multicolumn{1}{||c|}{$40$}\\\hline\hline
\multicolumn{4}{|l|}{$%
\begin{array}
[c]{c}%
\widehat{\theta}_{TW}\mathit{=}\frac{18}{30}\mathit{=}.6,\ \widehat{\theta
}_{TB}\mathit{=}\frac{7}{10}\mathit{=}.7
\end{array}
$}\\\hline
\end{tabular}
\end{tabular}
\]

\section{Revisiting the causal modeling `solution'\vspace*{-0.1in}}

Pearl's (2014) claims that the only way to resolve the paradox is to use
causal calculus:

\textsf{\textquotedblleft I am not aware of another condition that rules out
effect reversal with comparable assertiveness and generality, requiring only
that }$Z$\textsf{\ not be affected by our action, a requirement satisfied by
all treatment-independent covariates }$Z$\textsf{. Thus, it is hard, if not
impossible, to explain the surprise part of Simpson's reversal without
postulating that human intuition is governed by causal calculus together with
a persistent tendency to attribute causal interpretation to statistical
associations.\textquotedblright\ (p. 10)}

Viewing examples 2-5 from the misspecification perspective, however, lends
support to the Armistead's (2014) key argument:

\textsf{\textquotedblleft Whether causal or not, third variables can convey
critical information about a first-order relationship, study design, and
previously unobserved variables. Any conditioning on a nontrivial third
variable that produces Simpson's Paradox should be carefully examined before
either the aggregated or the disaggregated findings are accepted, regardless
of whether the third variable is thought to be causal. In some cases, neither
set of data is trustworthy; in others, both convey information of
value.\textquotedblright\ (p. 1)}

Indeed, in cases where the `third variable' represents an \textit{ordering} of
potential interest for the particular data, the only relevant criterion to
decide which orderings are relevant for the statistical analysis of the
particular data is the \textit{statistical adequacy} of the estimated
equations. That is, when two or more alternative orderings are potentially
relevant for a particular data set, one needs to test the statistical adequacy
of all three equations relative to each of these orderings before one could
draw reliable conclusions concerning how to resolve any apparently paradoxical
results. Where does this leave the Pearl (2014) claim quoted above?\vspace
*{-0.15in}

\subsection{Statistical vs. substantive adequacy\vspace*{-0.1in}}

Cartwright (1979) rightly points out that reliance on regularities and
frequencies for statistical inference purposes is not sufficient for
representing substantively meaningful causal relations. On the other hand,
imposing causal relations that belie the chance regularities in the data would
only give rise to untrustworthy inference results. While the causal dimension
remains an important component in delineating the issues raised by Simpson's
paradox, it is not the only relevant, or even the most important, dimension in
unraveling the puzzle. Indeed, the suggestion that in cases where the third
variable (ordering of interest) is noncausal one should accept the results
based on the aggregated data (Pearl, 2009), is called into question by
examples 2-5. This is because when the model estimated using the aggregated
data is statistically misspecified, the causal inference results pertaining to
conditional independence are likely to be untrustworthy. One way or another,
the modeler needs to account for the statistical information not accounted for
by the original statistical model, with a view to ensure the trustworthiness
of the ensuring statistical results.

It is interesting to note that Yule (1926) considered the third variable
causal explanation, but questioned its value as a general `solution' to the problem:

\textsf{\textquotedblleft Now it has been said that to interpret such
correlations as implying causation is to ignore the common influence of the
time-factor. While there is a sense -- a special and definite sense -- in
which this may perhaps be said to cover the explanation, as will appear in the
sequel, to my own mind the phrase has never been intellectually
satisfying.\textquotedblright\ (p. 4)}

A crucial issue that needs to be addressed by the causal explanation is that
conditioning on a third variable is not as straightforward as adherents to
this `explanation' of Simpson's paradox would have us believe. In practice,
the question whether a particular variable $Z_{t}$ constitutes a confounder is
not just a matter of testing whether $Z_{t}$ relates to $Y_{t}$ and $X_{t} $
the right way; see Pearl (2009), Spirtes et al. (2000). Before such testing
can even begin, one needs to test for the statistical adequacy of the
estimated model with respect to a relevant ordering. Although `time' is the
obvious ordering for time series, it is no different than other deterministic
orderings for cross-section data such as `gender', marital status, age,
geographical position, etc.; only the scale of measurement differs. When the
original model is statistically misspecified, it needs to be respecified with
a view to secure statistical adequacy. Often one can restore statistical
adequacy using generic terms relating to that ordering. To secure substantive
adequacy, however, one needs to replace such generic terms with\ proper
explanatory random variables without foregoing the statistical adequacy. The
latter ensures the reliability of testing whether $Z_{t}$ is a confounder or
not; see Spanos (2006b).

Yule (1926) considered `time' as a third variable and expressed his misgivings:

\textsf{\textquotedblleft I cannot regard time \textit{per se} as a causal
factor; and the words only suggest that there is some third quantity varying
with the time to which the changes in both the observed variables are
due.\textquotedblright\ (p. 4)}

Viewing his comment from the vantage point of today's probabilistic
perspective, the proposal to `condition' on a third variable raises technical
issues, since the conditional distribution, defined by:$\vspace*{-0.07in}$%
\[%
\begin{array}
[c]{c}%
f(y_{t}\mathit{\mid}x_{t},d_{t};\mathbf{\varphi})\mathit{=}\frac
{f(y_{t}\mathit{,}x_{t},d_{t};\mathbf{\psi})}{f(d_{t};\mathbf{\phi}%
)},\ \forall y_{t}\mathit{\in}\mathbb{R}_{Y},
\end{array}
\vspace*{-0.07in}%
\]
makes no probabilistic sense when $d_{t}$ is a \textit{determinist ordering}
(variable) such as time; see Williams (1991). This issue arises more clearly
in cases where the ordering was deemed potentially important after the data
have been collected, such as having plants grow short or tall, blood pressure
being high or low, black or white plants, etc.; see Armistead (2014). How does
one bridge the gap between a deterministic ordering of interest and
conditioning on a third random variable related to that ordering?

\textbf{Separating modeling from inference}. The statistical misspecification
perspective suggests that to ensure the reliability of inference one needs to
separate the initial stages of \textit{specification} (initial model
selection) \textit{misspecification testing} and \textit{respecification},
from \textit{inference} proper. The latter includes testing for
\textit{substantive adequacy}, such as attributing causality to statistical
associations. In practice, this requires focusing first on the ordering(s) of
interest that could potentially reveal statistical misspecifications that
pertain to dependence and heterogeneity uncovered by misspecification testing.
The next step is to respecify the initial model with a view to account for the
statistical information revealed by the misspecification testing. This is
usually achieved by employing \textit{generic} terms, such as shifts, trends
and lags, to `capture' such forms of systematic statistical information. Once
statistical adequacy is secured one can then proceed to `model' such
information by replacing the generic terms with appropriate explanatory
variables with a view to improve the \textit{substantive adequacy} without
forgoing the statistical adequacy. This is because a third degree trend
polynomial might capture the mean heterogeneity in the data to ensure the
statistical reliability of inference, but from the substantive perspective it
represents ignorance. Replacing the trend polynomial with explanatory
variables without forsaking statistical adequacy will add to our understanding
of the phenomenon of interest; see Spanos (2010).

Viewing the problem from a broader perspective, the primary reason for the
untrustworthiness is that the question of probing for the nature of any causal
connections pertains to \textit{substantive}, and not \textit{statistical
adequacy}, even though the distinction between the two might not always be
clear cut or obvious; see Spanos (2010). This distinction is crucial because
any attempt to probe for substantive adequacy, including causal connections,
before securing statistical adequacy\ of the assumed statistical model is
likely to give rise to unreliable results. To avoid this problem of unreliable
inferences, one needs to establish the statistical adequacy of the original
model first before probing for any form of substantive adequacy, such as
attributing a causal interpretation to statistical associations. These include
probing for the appropriateness of a particular confounder or choosing between
different potential confounders; see Spanos (2006b) for an extensive discussion.

This distinction is crucial in differentiating between \textit{statistically}
and \textit{substantively} `spurious' inferential results. Unfortunately, in
the statistics and philosophy of science literatures the term `spurious' is
often used to describe the latter; see Blyth (1972). What is often
insufficiently appreciated is that one needs to establish first that there is
a statistically trustworthy statistical association, before attempting to
explain it away as substantively spurious.

The statistical misspecification perspective also calls into question certain
philosophical discussions of Simpson's paradox that focus primarily on the
`numbers' associated with the relevant probabilities/associations as in the
case of example 1. A typical representation of Simpson's paradox in terms of
events $A,B,C$ is:$\vspace*{-0.07in}$%
\begin{equation}%
\begin{array}
[c]{c}%
\begin{tabular}
[c]{l}%
$P(A\mathit{\mid}B)\mathit{<}P(A\mathit{\mid}\lnot B),\text{ but}\medskip$\\
$P(A\mathit{\mid}B,C)\mathit{>}P(A\mathit{\mid}\lnot B,C)$ and$\medskip$\\
$P(A\mathit{\mid}B,\lnot C)\mathit{>}P(A\mathit{\mid}\lnot B,\lnot C),$%
\end{tabular}
\end{array}
\vspace*{-0.07in}\label{ph}%
\end{equation}
where `$\lnot$' denotes the `negation' operator. Malinas and Bigelow (2016)
illustrate (\ref{ph}) using made up numbers that satisfy the above
inequalities, and describe the source of the paradox as follows:

\textsf{\textquotedblleft The applications arise from the close connections
between proportions, percentages, probabilities, and their representations as
fractions.\textquotedblright\ (p. 1)}

\hspace*{-0.25in}They proceed to claim that their artificial illustration
provides a way to explain an empirical example from Cohen and Nagel (1934)
concerning death rates in 1910 from tuberculosis in Richmond, Virginia and New
York city. As argued above, however, in the case of observed data some of the
`numbers' used in such arguments might be statistically untrustworthy,
undermining the soundness of the logical argument in (\ref{ph}). Indeed,
oversimplifications of the form (\ref{ph}), contribute to the perpetuation of
the misconceptions beleaguering the paradox.\vspace*{-0.15in}

\section{Summary and conclusions\vspace*{-0.1in}}

What is often insufficiently appreciated in statistical modeling and inference
is that the inference propositions (optimal estimators, tests, and predictors
and their sampling distributions) depend crucially on the validity of the
probabilistic assumptions one imposes on the data. The totality of these
assumptions comprise the underlying statistical model, which is used to define
the distribution of the sample and the likelihood function. If the statistical
model is misspecified, in the sense that any of its assumptions are invalid
for the particular data, the reliability of inference based on such a model is
usually undermined, giving rise to untrustworthy evidence.

The paper revisited Simpson's paradox using the statistical misspecification
perspective with a view to shed light on several silent features of the
paradox. Using this perspective, it was argued that the key to unraveling the
various counterintuitive results associated with this paradox is to formalize
the vague notion of `spurious' inference results into `statistically
untrustworthy' results which can be evidenced using misspecification testing.
This enables one to distinguish between two different cases of the paradox as
it relates to the reversal of statistical associations. Case 1, where the
reversal is statistically trustworthy because the underlying statistical
models are statistically adequate (example 1). Case 2, where the apparent
reversal is statistically untrustworthy due to statistical misspecification
(examples 2-5). The real issue is whether the inference results pertaining to
statistical associations are statistically trustworthy or not, and the key
criterion to appraise that is statistical adequacy. Hence, the statistical
misspecification perspective puzzles out Simpson's paradox because in both
cases there is nothing counterintuitive to explain.

The statistical misspecification perspective is also used to revisit the
causal dimension of the paradox by distinguishing between statistical and
substantive inadequacy (spuriousness). To ensure the reliability of any
inferences relating to testing whether a third variable constitutes a
confounder, requires that the underlying statistical model is statistically
adequate. This is particularly problematic for the causal resolution of the
paradox when the third variable is related to a relevant ordering of interest
which is revealed after the data are collected. In such cases one needs to
account for any departures from the model assumptions as they relate to the
ordering in question, and replace the generic terms used to capture the
neglected statistical information with substantively meaningful explanatory
variables.$\vspace*{-0.15in}$

\newpage

\section{Appendix: The Linear Regression\ - implicit statistical
parameterizations$\vspace*{-0.1in}$}

The traditional specification of the LR model takes the form:%
\[%
\begin{array}
[c]{c}%
y_{t}\mathit{=}\beta_{0}+\mathbf{\beta}_{1}^{\top}\mathbf{x}_{t}%
\mathit{+}u_{t},\medskip\\
\text{\lbrack i]\ }E(u_{t}\mathit{\mid}\mathbf{X}_{t}\mathit{=}\mathbf{x}%
_{t})\mathit{=}0,\ \text{[ii]\ }E(u_{t}^{2}\mathit{\mid}\mathbf{X}%
_{t}\mathit{=}\mathbf{x}_{t})\mathit{=}\sigma^{2},\medskip\\
\text{\lbrack iii] }E(u_{t}u_{s}\mathit{\mid}\mathbf{X}_{t}\mathit{=}%
\mathbf{x}_{t})\mathit{=}0,\ \text{[iv] }u_{t}\mathit{\backsim}%
\text{\textsf{N}}(.,.),\ t\mathit{\in}\mathbb{N}.
\end{array}
\]

\begin{theorem}
Assumptions [i]-[iii] relating to the first two moments of the conditional
distribution $f(u_{t}\mathit{\mid}\mathbf{x}_{t};\mathbf{\theta}),$ imply that
the model parameters $\mathbf{\theta}$:\textbf{=}$(\beta_{0},\mathbf{\beta
}_{1},\sigma^{2})$ have the following statistical parameterizations in terms
of the primary parameters of the joint distribution $f(y_{t},\mathbf{x}%
_{t};\mathbf{\phi})$, $\mathbf{\phi\mathit{:=}}(E(y_{t}),\ E(\mathbf{X}%
_{t}),\ Cov(\mathbf{X}_{t}),\ Cov(\mathbf{X}_{t},y_{t}))$:
\begin{equation}%
\begin{tabular}
[c]{ll}%
$\beta_{0}\mathbf{=}E(y_{t})\mathit{-}\mathbf{\beta}_{1}^{\top}E(\mathbf{X}%
_{t}),$ & $\mathbf{\beta}_{1}\mathit{=}\mathbf{[}Cov(\mathbf{X}_{t}%
)]^{-1}Cov(\mathbf{X}_{t},y_{t}),\medskip$\\
\multicolumn{2}{l}{$\sigma^{2}\mathbf{=}Var(y_{t})\mathit{-}Cov(\mathbf{X}%
_{t},y_{t})^{\top}[Cov(\mathbf{X}_{t})]^{-1}Cov(\mathbf{X}_{t},y_{t})$}%
\end{tabular}
\label{p}%
\end{equation}

\end{theorem}

\begin{proof}
Assumption [i] implies that:$\vspace*{-0.07in}$%
\begin{equation}%
\begin{array}
[c]{c}%
E(u_{t}\mathit{\mid}\mathbf{X}_{t}\mathit{=}\mathbf{x}_{t})\mathbf{=}%
0\Leftrightarrow E(y_{t}\mathit{\mid}\mathbf{X}_{t}\mathit{=}\mathbf{x}%
_{t})\mathbf{=}\beta_{0}+\mathbf{\beta}_{1}^{\top}\mathbf{x}_{t}.
\end{array}
\label{ce}%
\end{equation}
The law of iterated expectations (Williams, 1991): $%
\begin{array}
[c]{c}%
\left(  E\left[  E(Y\mathit{\mid}\sigma(\mathbf{X}))\right]  \right)
\mathbf{=}E(Y),
\end{array}
$\newline where $\sigma(\mathbf{X})$ denotes the sigma-field generated by
$\mathbf{X},$ implies that:%
\[
E\left[  E(y_{t}\mathit{\mid}\sigma(\mathbf{X}_{t}))\right]  \mathbf{=}%
E(y_{t})\mathbf{=}\beta_{0}\mathit{+}\mathbf{\beta}_{1}^{\top}E(\mathbf{X}%
_{t})\Longrightarrow\beta_{0}\mathit{=}E(y_{t})-\mathbf{\beta}_{1}^{\top
}E(\mathbf{X}_{t})
\]
Substituting $\beta_{0}$ back into $y_{t}\mathit{=}\beta_{0}+\mathbf{\beta
}_{1}^{\top}\mathbf{x}_{t}\mathit{+}u_{t}$ yields:%
\[
\left[  y_{t}-E(y_{t})\right]  \mathit{=}\mathbf{\beta}_{1}^{\top}\left[
\mathbf{X}_{t}-E(\mathbf{X}_{t})\right]  +u_{t}.
\]
Post-multiplying both sides by $\left[  \mathbf{X}_{t}\mathit{-}%
E(\mathbf{X}_{t})\right]  ^{\top}$ and taking expectations yields:%
\[%
\begin{array}
[c]{cl}%
Cov(y_{t},\mathbf{X}_{t}) & \mathit{:=}E\left(  \left[  y_{t}\mathit{-}%
E(y_{t})\right]  \left[  \mathbf{X}_{t}\mathit{-}E(\mathbf{X}_{t})\right]
^{\top}\right)  \mathit{=}\\
& \mathit{=}\mathbf{\beta}_{1}^{\top}\left[  \mathbf{X}_{t}\mathit{-}%
E(\mathbf{X}_{t})\right]  \left[  \mathbf{X}_{t}\mathit{-}E(\mathbf{X}%
_{t})\right]  ^{\top}\mathit{+}E(u_{t}\left[  \mathbf{X}_{t}\mathit{-}%
E(\mathbf{X}_{t})\right]  ^{\top}).
\end{array}
\]
Since, the last term is zero: $%
\begin{array}
[c]{c}%
E(\mathbf{X}_{t}^{\top}u_{t})\mathit{=}E\left[  E(u_{t}\mathit{\mid}%
\sigma(\mathbf{X}_{t}))\right]  \mathit{=}0,
\end{array}
\smallskip\newline$it follows that: $\mathbf{\beta}_{1}\mathit{=}%
\mathbf{[}Cov(\mathbf{X}_{t})]^{-1}Cov(\mathbf{X}_{t},y_{t})$. 

In the case of $\sigma^{2}$ we use a theorem analogous to the \textit{lie} for
the variance (Williams, 1991):\vspace*{-0.1in}%
\[%
\begin{array}
[c]{c}%
Var(y_{t})=E\left[  Var(y_{t}\mathit{\mid}\sigma(\mathbf{X}_{t}))\right]
+Var\left[  E(y_{t}\mathit{\mid}\sigma(\mathbf{X}_{t}))\right]  ,
\end{array}
\]
where, by definition $E\left[  Var(y_{t}\mathit{\mid}\sigma(\mathbf{X}%
_{t}))\right]  \mathbf{=}\sigma^{2}.$ The mean deviation of (\ref{ce})
is:\vspace*{-0.1in}%
\[%
\begin{array}
[c]{c}%
\left[  \beta_{0}\mathit{+}\mathbf{\beta}_{1}^{\top}\mathbf{X}_{t}\right]
-E\left(  \left[  \beta_{0}\mathit{+}\mathbf{\beta}_{1}^{\top}\mathbf{X}%
_{t}\right]  \right)  \mathit{=}\mathbf{\beta}_{1}^{\top}\left[
\mathbf{X}_{t}-E(\mathbf{X}_{t})\right]  ,
\end{array}
\]
and thus, by definition:\vspace*{-0.1in}%
\[%
\begin{array}
[c]{c}%
Var\left[  E(y_{t}\mathit{\mid}\mathbf{X}_{t})\right]  \mathbf{=}E\left[
\mathbf{\beta}_{1}^{\top}\left[  \mathbf{X}_{t}\mathit{-}E(\mathbf{X}%
_{t})\right]  \left[  \mathbf{X}_{t}\mathit{-}E(\mathbf{X}_{t})\right]
^{\top}\mathbf{\beta}_{1}\right]  \mathit{=}\mathbf{\beta}_{1}^{\top}\left[
Cov(\mathbf{X}_{t})\right]  \mathbf{\beta}_{1}.
\end{array}
\vspace*{-0.1in}%
\]
From this, it follows that:\vspace*{-0.1in}
\[%
\begin{array}
[c]{c}%
Var(y_{t})\mathbf{=}\sigma^{2}\mathit{+}\mathbf{\beta}_{1}^{\top}\left[
Cov(\mathbf{X}_{t})\right]  \mathbf{\beta}_{1}\Longrightarrow\sigma
^{2}\mathit{=}Var(y_{t})\mathit{-}\mathbf{\beta}_{1}^{\top}\left[
Cov(\mathbf{X}_{t})\right]  \mathbf{\beta}_{1},
\end{array}
\vspace*{-0.1in}%
\]
which yields the parameterization in (\ref{p}).
\end{proof}


\begin{thebibliography}{99}                                                                                               %
\bibitem {}Armistead, T.W. (2014), \textquotedblleft Resurrecting the Third
Variable: A Critique of Pearl's Causal Analysis of Simpson's
Paradox\textquotedblright, \textit{The American Statistician}, 68:
1-7.\vspace*{-0.1in}

\bibitem {}Bickel, P.J., E.A. Hammel, and J.W. O'Connell (1975),
\textquotedblleft Sex Bias in Graduate Admissions: Data From
Berkeley,\textquotedblright\ \textit{Science}, 187, 398--404.\vspace*{-0.1in}

\bibitem {}Bishop, Y.M., S.E. Fienberg, P.W. Holland (1975), \textit{Discrete
Multivariate Analysis: Theory and Practice}, MIT Press, Cambridge,
MA.\vspace*{-0.1in}

\bibitem {}Blyth, C.R., (1972), \textquotedblleft On Simpson's Paradox and the
Sure Thing Principle\textquotedblright, \textit{Journal of the American
Statistical Association}, 67: 364--366.\vspace*{-0.1in}

\bibitem {}Carlton, A. G. (1946), \textquotedblleft Estimating the Parameters
of a Rectangular Distribution,\textquotedblright\ \textit{The Annals of
Mathematical Statistics}, \textbf{17}, 355-358.\vspace*{-0.1in}

\bibitem {}Cartwright, N., (1979), \textquotedblleft Causal laws and effective
strategies\textquotedblright, \textit{No\^{u}s}, 13 (4): 419--437.\vspace
*{-0.1in}

\bibitem {}Cohen, M., and Nagel, E. (1934), \textit{An Introduction to Logic
and the Scientific Method}, New York: Harcourt, Brace and Company.\vspace
*{-0.1in}

\bibitem {}Fisher, R. A. (1922), \textquotedblleft On the mathematical
foundations of theoretical statistics\textquotedblright, \textit{Philosophical
Transactions of the Royal Society A}, \textbf{222}: 309-368.\vspace*{-0.1in}

\bibitem {}Greene, W.H. (2011), \textit{Econometric Analysis}, 7th ed., New
Jersey: Prentice Hall.\vspace*{-0.12in}

\bibitem {}Kolmogorov,\ A. N. (1933), \textit{Foundations of the theory of
Probability}, 2nd English edition, Chelsea Publishing Co. NY.\vspace*{-0.12in}

\bibitem {}Le Cam, L. (1986), \textit{Asymptotic Methods in Statistical
Decision Theory}, SpringerVerlag, NY.\vspace*{-0.12in}

\bibitem {}Lindley, D.V., and M.R. Novick (1981), \textquotedblleft The role
of exchangeability in inference\textquotedblright, \textit{Journal of the
American Statistical Association}, 9: 45--58.\vspace*{-0.1in}

\bibitem {}Malinas, G. and J. Bigelow (2016), \textquotedblleft Simpson's
Paradox\textquotedblright, The Stanford Encyclopedia of Philosophy (Summer
2016 Edition), Edward N. Zalta (ed.), forthcoming URL =
$<$%
http://plato.stanford.edu/archives/sum2016/entries/paradox-simpson/%
$>$%
.\vspace*{-0.1in}

\bibitem {}Pearl, J. (2009), \textit{Causality: Models, Reasoning, and
Inference} (2nd ed.), Cambridge University Press, NY.\vspace*{-0.1in}

\bibitem {}Pearl, J. (2011), \textquotedblleft Why there is no statistical
test for confounding, why many think there is, and why they are almost
right\textquotedblright, Department of Statistics, UCLA.\vspace*{-0.1in}

\bibitem {}Pearl, J. (2014), \textquotedblleft Comment: Understanding
Simpson's Paradox\textquotedblright, \textit{The American Statistician},
68(1), 8-13.\vspace*{-0.1in}

\bibitem {}Pearson, K. (1896), \textquotedblleft Mathematical Contributions to
the Theory of Evolution--On a Form of Spurious Correlation Which May Arise
When Indices Are Used in the Measurement of Organs\textquotedblright,
\textit{Proceedings of the Royal Society of London}, 60(359-367),
489-498.\vspace*{-0.1in}

\bibitem {}Samuels, M.L. (1993), \textquotedblleft Simpson's Paradox and
Related Phenomena\textquotedblright, \textit{Journal of the American
Statistical Association}, 88:421, 81-88.\vspace*{-0.1in}

\bibitem {}Simpson, E.H. (1951), \textquotedblleft The interpretation of
interaction in contingency tables\textquotedblright, \textit{Journal of the
Royal Statistical Society}. Series B (Methodological), 238-241.\vspace
*{-0.1in}

\bibitem {}Sober, E. (2001), \textquotedblleft Venetian sea levels, British
bread prices, and the principle of the common cause\textquotedblright,
\textit{The British Journal for the Philosophy of Science}, 52:
331-346.\vspace*{-0.1in}

\bibitem {}Spanos, A., (1986), \textit{Statistical Foundations of Econometric
Modelling}, Cambridge University Press, Cambridge.\vspace*{-0.1in}

\bibitem {}Spanos, A. (1999), \textit{Introduction to Probability Theory and
Statistical Inference}, Cambridge University Press.\vspace*{-0.1in}

\bibitem {}Spanos, A. (2006a), \textquotedblleft Where Do Statistical Models
Come From? Revisiting the Problem of Specification,\textquotedblright\ pp.
98-119 in \textit{Optimality: The Second Erich L. Lehmann Symposium}, edited
by J. Rojo, Lecture Notes-Monograph Series, vol. 49, Institute of Mathematical
Statistics.\vspace*{-0.12in}

\bibitem {}Spanos, A. (2006b), \textquotedblleft Revisiting the Omitted
Variables Argument: Substantive vs. Statistical Adequacy,\textquotedblright%
\ \textit{Journal of Economic Methodology}, 13: 179-218.\vspace*{-0.12in}

\bibitem {}Spanos, A. (2010), \textquotedblleft Statistical Adequacy and the
Trustworthiness of Empirical Evidence: Statistical vs. Substantive
Information,\textquotedblright\ \textit{Economic Modelling}, 27: 1436--1452.
\vspace*{-0.12in}

\bibitem {}Spanos, A. (2015), \textquotedblleft Statistical Mis-Specification
Testing in Retrospect and Prospect\textquotedblright, working paper, Virginia
Tech.\vspace*{-0.12in}

\bibitem {}Spanos, A. and A. McGuirk (2001), \textquotedblleft The Model
Specification Problem from a Probabilistic Reduction
Perspective,\textquotedblright\ \textit{Journal of the American Agricultural
Association}, \textbf{83}: 1168-1176. \vspace*{-0.12in}

\bibitem {}Spanos, A. and A. McGuirk (2002), \textquotedblleft The Problem of
Near-Multicollinearity Revisited: erratic vs. systematic
volatility,\textquotedblright\ \textit{Journal of Econometrics}, \textbf{108}:
365-393.\vspace*{-0.12in}

\bibitem {}Spirtes, P., C.N. Glymour, R. Scheines (2000), \textit{Causation,
prediction, and search}. MIT press, MA.\vspace*{-0.1in}

\bibitem {}Stigler, S. M. (1980), \textquotedblleft Stigler's Law of
Eponymy\textquotedblright, \textit{Transactions of the New York Academy of
Sciences}, 39: 147-157.\vspace*{-0.1in}

\bibitem {}Wasserman, L. (2004), \textit{All of Statistics: A Concise Course
in Statistical Inference}, Springer, NY.\vspace*{-0.1in}

\bibitem {}Williams, D. (1991), \textit{Probability with Martingales},
Cambridge University Press, Cambridge.\vspace{-0.12in}

\bibitem {}Yule, G.U. (1903), \textquotedblleft Notes on the theory of
association of attributes in statistics\textquotedblright, \textit{Biometrika}%
, 2: 121-134.\vspace{-0.12in}

\bibitem {}Yule, G.U. (1909), \textquotedblleft The applications of the method
of correlation to social and economic statistics\textquotedblright, Journal of
the Royal Statistical Society, 72: 721-730.\vspace*{-0.12in}

\bibitem {}Yule, G. U. (1910), \textquotedblleft On the interpretation of
correlations between indices or ratios\textquotedblright, \textit{Journal of
the Royal Statistical Society}, 73: 644-647.\vspace*{-0.12in}

\bibitem {}Yule, G. U. (1921), \textquotedblleft On the time-correlation
problem, with especial reference to the variate-difference correlation
method\textquotedblright, \textit{Journal of the Royal Statistical Society},
84: 497-537.\vspace*{-0.12in}

\bibitem {}Yule, G.U. (1926), \textquotedblleft Why do we sometimes get
nonsense-correlations between Time-Series?--a study in sampling and the nature
of time-series, \textit{Journal of the Royal Statistical Society}, 89:
1-64.\vspace*{-0.12in}
\end{thebibliography}
\end{document}